\documentclass[%
 reprint,
superscriptaddress,
preprintnumbers,
 amsmath,amssymb,
 aps,prl,
footinbib,
notitlepage
]{revtex4-1}

\usepackage{color} 
\usepackage{graphicx}
\usepackage{dcolumn}
\usepackage{bm}

\usepackage{verbatim}
\usepackage{appendix}
\usepackage{amsthm}
\usepackage{float}
\theoremstyle{plain}
\newtheorem{thm}{Theorem}

\theoremstyle{definition}

\begin{document}
\preprint{MIT-CTP/4803}
\title{Resource Destroying Maps}
\author{Zi-Wen Liu}\email{zwliu@mit.edu}
\affiliation{Center for Theoretical Physics and Department of Physics, Massachusetts Institute of Technology, Cambridge, Massachusetts 02139, USA}
\author{Xueyuan Hu}
\affiliation{School of Information Science and Engineering, and Shandong Provincial Key Laboratory\\ of Laser Technology and Application, Shandong University, Jinan 250100, China}
\author{Seth Lloyd}
\affiliation{Department of Mechanical Engineering, Massachusetts Institute of Technology, Cambridge, Massachusetts 02139, USA}

\date{\today}

\begin{abstract}
  Resource theory is a widely applicable framework for analyzing the
physical resources required for given tasks, such as computation, 
communication, and energy extraction. 
In this Letter,
we propose a general scheme for analyzing resource 
theories based on resource destroying maps, which leave 
resource-free states unchanged but erase the resource stored in 
all other states. 
We introduce a group of general conditions that determine whether a quantum operation exhibits typical resource-free properties in relation to a given resource destroying map.
Our theory reveals fundamental 
connections among basic elements of resource theories, in
particular, free states, free operations, and resource measures.
In particular, we define a class of simple resource measures that can
be calculated without optimization, and that are monotone nonincreasing 
under operations that commute with the resource destroying map. 
We apply our theory to the resources of coherence and quantum correlations
(e.g., discord), two prominent features of nonclassicality. 
  
\end{abstract}

\maketitle

\emph{Introduction.|}Resource theory originates from the 
observation that certain properties of physical systems become valuable 
resources when the operations that can be performed are restricted so that such properties are hard to create. 
A prototypical example of such a property is quantum entanglement 
\cite{entrev,virmani}, which becomes a key resource for many 
quantum information processing tasks, when one is restricted to local 
operations and classical communication (LOCC). 
The framework of resource theory has been applied to various other 
concepts in quantum information, such as purity \cite{purityrt}, 
magic states \cite{magic}, and coherence 
\cite{baumgratz,winteryang}, and to broader areas, such as 
asymmetry \cite{noether} and thermodynamics \cite{thermo}. 

Theories of different resources share a similar structure. In general, 
quantum resource theories contain three basic elements: free 
states, free (allowed) operations, and resource measures (monotones). 
 These elements are closely related to one another. 
For example, free operations should not be able to create 
resource from free states, and resource measures are expected to be monotone nonincreasing under free operations.
In recent years, considerable effort has been devoted to developing a 
unified framework of resource theories \cite{bg,qrtq,Coecke}. 
In particular, Ref.\ \cite{bg} studies the general case where 
the set of free operations is maximal, i.e., all (asymptotically) 
resource nongenerating operations are allowed, and when the resource 
satisfies several postulates (e.g., the set of free states is convex).

Some key aspects of resource theories are not addressed 
by existing frameworks, however. For example, characterizing a proper set of free operations is frequently a major difficulty in establishing a resource theory, and we do not yet have general principles and understandings for nonmaximal theories. Indeed, a successful resource theory is usually specified by physical restrictions on the set of allowed operations: LOCC and thermal operations \cite{janzing,nanothermo,thermo} are prominent examples. But such restrictions are often stronger than merely nongenerating, and may lead to mathematical difficulties in characterizing and calculating monotones.  Moreover, existing results do not apply to some 
resources, such as discord, where the set of free states is nonconvex. 

In this Letter, we introduce a simple but 
universally applicable theory of resource-free properties of quantum 
operations that addresses these issues. 
Our theory is based on the notion of resource destroying maps: for a 
given resource, a resource destroying map leaves free states unchanged, 
but destroys the resource otherwise. Key features of resource destroying 
maps are discussed.  For example, an immediate observation is that a 
resource destroying map is not linear (thus cannot be represented by a 
quantum channel) if the set of free states is nonconvex. As will be seen, many important properties of our framework sharply contrast linear resource destroying maps with nonlinear ones.
We demonstrate that the concept of resource destroying maps
helps unify and simplify the analysis of resource theories,
allowing us to determine whether a quantum operation exhibits a group of
fundamental resource-free properties, in addition to nongenerating. 
A basic result of our theory is that any contractive distance 
between a state and its resource-free version is monotone 
nonincreasing under all such operations.
Finally, we apply the framework of resource destroying maps to 
coherence and discord. In particular, we find that the theory of discord, which is poorly understood in terms of resource theory (largely due to its nonconvexity), can exhibit a simple structure in this framework. 
Moreover, the analysis of discord helps illustrate several peculiar properties of nonlinear resource destroying maps.

\emph{Resource destroying maps.|}Here we formally define the notion 
of \emph{resource destroying maps}, the key concept of our theory. 
Let $F$ be the set of free states for a certain theory. For all 
input states $\rho$, a resource destroying map $\lambda$ satisfies 
the following requirements: 
(i) resource destroying: if $\rho\not\in F$, $\lambda(\rho)\in F$; 
(ii) nonresource fixing: if $\rho\in F$, $\lambda(\rho)=\rho$. 
In other words, a resource destroying map outputs 
a free state if the input is not free, and leaves the input 
unchanged otherwise. 
The resource destroying map characterizes the resource-free
space: 
$F$ consists precisely of the fixed points of $\lambda$.
Resource destroying maps are idempotent due to (ii). They are also surjections onto 
codomain $F$ since every free state is a preimage of itself.
It is helpful to draw an analogy between the structure of resource 
destruction and a fiber bundle: $\lambda$ defines a bundle projection 
onto $F$. Call a nonfree state a \emph{parent state} of its image 
free state. Then each free state defines a \emph{family} consisting 
of corresponding parent states (the fiber) and the free state itself.

Note that a resource destroying map does not have to be completely
positive or linear, and can be highly nonuniform. 
However, we are mostly interested in the physically motivated maps, 
usually with simple descriptions that work universally for all inputs. 
For example, the simplest case is when the resource destroying map
can be represented by a quantum channel. 
However, it can be shown that $\lambda$ cannot be linear (thus not a channel) when $F$ is nonconvex. (See Supplemental
Material \footnote{See Supplemental Material [url], which includes Refs.\ \cite{vedralplenio,time,refrmp,asymgour,2016arXiv161007504B,wilde,watrous,guohou,shorppt}} for details.)
In addition, for theories of correlations among multiple parties, 
local resource destroying maps cannot be a channel either. 
Notably, entanglement breaking channels \cite{entbreak} do not 
necessarily leave separable (unentangled) states unchanged, 
and so are not entanglement destroying maps.   Consider uncorrelated 
states: the channel that stabilizes all local states can only be the 
identity, which does not destroy resource. 
Necessary and sufficient conditions for the existence of resource destroying channels are recently investigated in Ref.\ \cite{singleshot}.

For many theories, a simple resource destroying map is easy to identify. 
For example, complete dephasing in the
preferred basis is an obvious coherence destroying map;
Haar (uniform) twirling over the group $G$ is a $G$-asymmetry destroying map \cite{clock}.
For discord-type quantum correlations, the resource destroying 
map cannot be a channel (whether local or not) since discord-free 
(classically correlated) states form a nonconvex set \cite{modi}, but
it can simply be a local measurement in an eigenbasis of the 
reduced density operator.
In the following, we use upper and lower case Greek letters to 
denote channels and general maps, respectively.

\emph{Resource-free conditions.|}\label{conditions}Now we are ready to introduce a group of general conditions with simple mathematical forms, based on resource destroying maps, which correspond to various typical resource-free properties of quantum operations. 

Consider a theory with resource destroying map $\lambda$. Let $\mathcal{E}$ be some quantum operation.
We start from
\begin{equation}
    \mathcal{E}\circ\lambda=\lambda\circ\mathcal{E}\circ\lambda,
\end{equation}
where $\circ$ is the composition of maps. An equivalent form of this condition is the following: $\mathcal{E}(\lambda(\rho))=\lambda(\mathcal{E}(\lambda(\rho)))$ for all $\rho$. 
Recall that only free states are fixed points of $\lambda$. This condition indicates that the output of $\mathcal{E}\circ\lambda$ is always a fixed point of $\lambda$, thus free. In other words, the set of free states is closed under $\mathcal{E}$.
So we call this condition the {\it nongenerating condition}, and, correspondingly, the operations satisfying this condition {\it resource nongenerating operations}. This is a necessary constraint on free operations, since any other operation can create resource, thus trivializing the theory. Theories that allow all such operations (under some assumptions including convexity) possess a common structure: they are reversible and have regularized relative entropy as the unique monotone asymptotically \cite{bg,qrtq}.

Next, we consider the following dual form of the nongenerating condition:
\begin{equation}
    \lambda\circ\mathcal{E}=\lambda\circ\mathcal{E}\circ\lambda.
\end{equation}
Think of the output of $\lambda$ as the free part of an input state.
This condition means that $\mathcal{E}$ cannot make use of the 
resource stored in any input to affect the free part. We call this condition the {\it nonactivating condition}.
An alternative interpretation is that 
such operations never break up a family: 
members of the same family must be mapped to the same target family (not necessarily the original one though). 
An illustration of the nongenerating and nonactivating conditions is given in Fig.\ \ref{cond}.
\begin{figure}
\includegraphics[width=0.56\columnwidth]{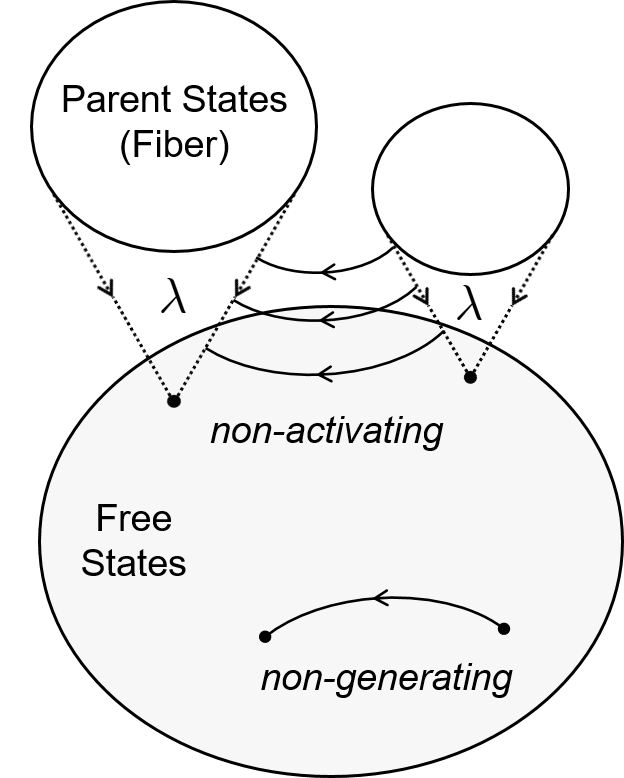}
\caption{\label{cond}An illustration of the resource-free conditions. The set of free states is closed under resource nongenerating operations. States belonging to the same family are mapped to the same target family by resource nonactivating operations.}
\end{figure}

In general, the nongenerating and nonactivating conditions can hold independently.
Because of the idempotence of $\lambda$, the sufficient and necessary condition for an
operation to be resource nongenerating and nonactivating 
simultaneously is that it commutes with $\lambda$:
\begin{equation}
    \lambda\circ\mathcal{E}=\mathcal{E}\circ\lambda.
\end{equation}
We call this condition the {\it commuting condition}.

Recall that a quantum operation $\mathcal{E}$ can be specified by 
Kraus decomposition $\mathcal{E}(\cdot)=\sum_\mu K_\mu\cdot K_\mu^\dagger$, 
where $\{K_\mu\}$ are Kraus operators satisfying 
$\sum_\mu K_\mu^\dagger K_\mu\leq I$. Each Kraus arm 
$\mathcal{E}_\mu(\cdot)\equiv K_\mu\cdot K_\mu^\dagger$ corresponds to a 
(unnormalized) generalized measurement outcome with probability 
${\rm tr}(K_\mu\cdot K_\mu^\dagger)$. In practice, one may want to require 
that the nongenerating, nonactivating, or commuting conditions be 
satisfied even when considering selective measurements, i.e., the 
outcome of the measurement is accessible. This leads to the following 
modification of each condition: there is some Kraus decomposition 
of $\mathcal{E}$ such that all $\mathcal{E}_\mu$ satisfies the condition. 
We call such counterparts \emph{selective conditions}. 
In other words, selective operations can be implemented by some 
POVM that exhibits corresponding resource-free properties, even if 
measurement outcomes are retained. Here we do not impose these 
conditions on every Kraus decomposition: typically, the relevant 
decomposition is specified by how we implement the operation, and this can be an overly strong requirement that places extra constraints irrelevant to the resource under study \footnote{In the context of coherence, this issue was touched upon by Ref.\ \cite{genuine}: all Kraus decompositions of Genuinely Incoherent Operations are incoherent, but they are not even capable of mapping an incoherent state to another as a result.}. 
We shall compare the strength of the original conditions and their selective counterparts in the next section.

For a given resource-free set $F$,
the definition of $\lambda$ is in general nonunique. 
Since $\lambda$ is surjective, 
the set of resource nongenerating operations is not affected by 
different choices of $\lambda$. In contrast, resource nonactivating operations and thus commuting operations can depend on the bundle structure specified by $\lambda$. 
These observations also hold for the selective version of each condition.
Explicit examples are given in the Supplemental Material \cite{Note1}. 

\emph{General properties.|}Here we introduce some typical features of our framework that hold generally in different theories. We shall see that some of these features manifestly contrast linear resource destroying maps with nonlinear ones. Denote the sets of resource nongenerating, nonactivating and commuting operations as $\bar{\mathbb{X}},\bar{\mathbb{X}}^\ast$ and $\mathbb{X}$, respectively, and their selective versions by subscript $s$. By definition, they satisfy $\mathbb{X}=\bar{\mathbb{X}}\cap\bar{\mathbb{X}}^\ast$ and $\mathbb{X}_s=\bar{\mathbb{X}}_s\cap\bar{\mathbb{X}}^\ast_s$.

For a theory with resource destroying channel $\Lambda$, one can easily construct these operations. Notice that $\Lambda\circ\Omega\in \bar{\mathbb{X}}$,
where $\Omega$ is an arbitrary operation, by the idempotence of 
$\Lambda$. Meanwhile, $\Lambda\circ\Omega$ belongs to 
$\bar{\mathbb{X}}^\ast$ only if $\Omega$ itself does. Similarly, 
$\Omega\circ\Lambda\in \bar{\mathbb{X}}^\ast$. 
Destroying the resource in both the input and output allows both conditions to be satisfied: $\Lambda\circ\Omega\circ\Lambda \in \mathbb{X}$. Selective operations can be constructed by similar procedures on each Kraus arm. Let $\{M_\mu\}$ be a Kraus decomposition of $\Omega$, and $\Omega_\mu(\cdot)\equiv M_\mu\cdot M_\mu^\dagger$ denote the action of each Kraus arm. It can be directly verified that each $\Lambda\circ \Omega_\mu$ specifies a resource nongenerating Kraus arm, i.e., $\sum_\mu\Lambda\circ \Omega_\mu\in\bar{\mathbb{X}}_s$. Similarly, $\sum_\mu\Omega_\mu\circ\Lambda\in\bar{\mathbb{X}}^\ast_s$ and $\sum_\mu\Lambda\circ\Omega_\mu\circ\Lambda\in\mathbb{X}_s$.

One may also ask if the resource-free properties hold for compositions and convex combinations. 
The answer is Yes for compositions for any $\lambda$.
For example, $\mathbb{X}$ is obviously closed under composition: 
given two operations $\mathcal{E}_1$ and $\mathcal{E}_2$ satisfying $\mathcal{E}_{1,2}\circ\lambda=\lambda\circ\mathcal{E}_{1,2}$ for some resource destroying map $\lambda$, it holds that $\mathcal{E}_2\circ\mathcal{E}_1$ is also a $\lambda$-commuting operation: by using the respective commuting conditions, we obtain $(\mathcal{E}_2\circ\mathcal{E}_1)\circ\lambda=\mathcal{E}_2\circ\lambda\circ\mathcal{E}_1=\lambda\circ(\mathcal{E}_2\circ\mathcal{E}_1)$. This conclusion also holds for $\bar{\mathbb{X}}$, $\bar{\mathbb{X}}^\ast$, and selective classes, which can be proven by similar arguments.
On the other hand, all classes are closed under convex 
combination when $\lambda$ is a linear map. 
Again, take the commuting condition as an example: 
$(p\mathcal{E}_1+(1-p)\mathcal{E}_2)\circ\lambda=p\mathcal{E}_1\circ\lambda+(1-p)\mathcal{E}_2\circ\lambda=p\lambda\circ\mathcal{E}_1+(1-p)\lambda\circ\mathcal{E}_2=\lambda\circ(p\mathcal{E}_1+(1-p)\mathcal{E}_2)$. Similar arguments work for other conditions. For nonlinear $\lambda$, however,
the last equality does not necessarily hold. 
For the same reason, when $\lambda$ is linear, selective conditions are stronger than their respective original versions (e.g., $\mathbb{X}_s\subset\mathbb{X}$), but otherwise this is not necessarily true.

We now show that the commuting condition plays a special role in the quantification of resources, a central theme of resource theories. 
The most basic property of a proper resource measure 
(a non-negative real function of states) is monotonicity under 
free operations: free operations should not be able to increase the 
amount of resource.
A natural type of measure is the minimal distance to the set of 
free states, where the distance is given by some function 
$D(\rho,\sigma)$ defined on two states $\rho$ and $\sigma$ that is 
contractive, i.e., obeys the data processing inequality 
$D(\Gamma(\rho),\Gamma(\sigma))\leq D(\rho,\sigma)$ for any 
operation $\Gamma$.   
Note that $D$ is not necessarily a metric. Nonsymmetric 
distances such as relative R\'enyi entropies are also valid choices of $D$. 
Formally, a distance measure of resource is given by  
$\mathfrak{D}(\rho):=\inf_{\sigma\in F}D(\rho,\sigma)$.
Monotonicity holds for such measures due to the minimization. However, such optimizations 
are often computationally hard. 
Now consider the following function:
\begin{equation}\label{simpled}
    \tilde{\mathfrak{D}}(\rho):=D(\rho,\lambda(\rho)).
\end{equation}
Because of the absence of minimization, $\tilde{\mathfrak{D}}(\rho)\geq \mathfrak{D}(\rho)$. However, if we restrict the set of allowed operations to $\mathbb{X}$, this measure also satisfies the monotonicity requirement: 
\begin{eqnarray}\label{mono}
\tilde{\mathfrak{D}}(\rho)&\geq& D(\Gamma(\rho),\Gamma(\lambda(\rho)))\nonumber\\&=&D(\Gamma(\rho),\lambda(\Gamma(\rho)))\equiv\tilde{\mathfrak{D}}(\Gamma(\rho)),
\end{eqnarray}
where the inequality follows from the contractivity of $D$. 
Therefore, for any resource theory with free operations satisfying the commuting condition, we have a class of computationally easy monotones which avoid optimizations (given that $\lambda$ is suitably defined). 
We should note that $\tilde{\mathfrak{D}}$ is not necessarily continuous everywhere when $\lambda$ is nonlinear, which requires more careful analysis in practice (as will be demonstrated for discord). 
The possibility of retaining measurement outcomes leads to the {\it selective monotonicity} condition---monotonicity under selective measurements on average. Following a similar argument as Eq.\ (\ref{mono}), a general result we can obtain at the moment is that $\tilde{\mathfrak{D}}$ obeys selective monotonicity under selective commuting operations, for a restricted class of $D$ including quantum relative entropy (details in the Supplemental Material \cite{Note1}).
Recall that, when $\lambda$ is linear, $\mathbb{X}_s\subset\mathbb{X}$: selective monotonicity is stronger than monotonicity; however, this is not necessarily the case when $\lambda$ is nonlinear.

\emph{Examples.|}We first focus on the theory of quantum coherence.   
Here, a basis of interest is specified, and density operators that are 
diagonal in this basis are incoherent (free). 
The study of coherence from a resource theory perspective has attracted a considerable amount of attention and effort in recent years. A few definitions of coherence-free operations stemmed from various perspectives are proposed and studied lately \cite{baumgratz,winteryang,speed,genuine,powerg,sio,speakable,examination,hung}, most of which can directly emerge from our framework as follows. 
 Complete dephasing in the preferred basis, denoted by $\Pi$, is a natural coherence-destroying map. Let $\bar{X}(\Pi)$ and $\bar{X}^\ast(\Pi)$ and $X(\Pi)$ be the sets of coherence nongenerating, nonactivating and commuting operations given by $\Pi$, respectively (an additional subscript $s$ for selective operations).
 $\bar{X}(\Pi)$ contains all coherence nongenerating operations, which are recently analyzed in Ref.\ \cite{hung}. 
Members of $\bar{X}^\ast(\Pi)$ cannot activate the coherence stored in the input in the sense that $\mathcal{E}(\cdot)$ and $\mathcal{E}\circ\Pi(\cdot)$ are always indistinguishable by measuring incoherent observables.
So $X(\Pi)$ contains operations that can neither create nor activate coherence.
In the preparation of this Letter, we became aware that 
these operations were very recently studied as dephasing-covariant operations in Refs.\ \cite{speakable,examination}. 
$\bar{X}_s(\Pi)$ and $X_s(\Pi)$ are respectively the sets of Incoherent Operations \cite{baumgratz} and Strictly Incoherent Operations \cite{sio}.
Detailed discussions of these classes and further comparisons to 
other relevant proposals of coherence-free operations are provided in 
the Supplemental Material \cite{Note1}. For any theory where the free operations 
belong to $X(\Pi)$, we know that $D(\cdot,\Pi(\cdot))$ for any 
contractive $D$ represents a coherence monotone. In comparison, 
monotonicity of some $D$ may fail if more operations are allowed. For example, not all relative R\'enyi entropies are monotone under $\bar{X}(\Pi)$ \cite{examination}.

Next, we consider discord \cite{zurek,hv}, the most general form of 
nonclassical correlations; see Ref.\ \cite{modi} for a comprehensive review.
Discord places a stronger constraint on free states than entanglement in 
the sense that it can exist in separable states. Discord 
has been shown to be the underlying resource for various tasks 
\cite{dqc1,rsp,qc,qcrypt}. However, a formal treatment of discord in 
the resource theory framework (e.g., transformation rules) remains elusive, mostly because our understanding of discord-free operations is 
limited, and most existing general results for resource theory do 
not directly apply to discord, due to its nonconvexity. 
 Here, we focus on the one-sided discord as measured on subsystem $A$ of a bipartite state $\rho_{AB}$, and local operations acting on the same subsystem.  The ideas can be generalized 
to nonlocal operations and multipartite cases.
 A state is regarded as discord-free if there exist local rank-one 
projective measurements that do not perturb the joint state. 
Such states take the form 
$
    \rho_{AB}=\sum_i p_i  |i\rangle_A\langle i|\otimes\rho_B^i,
$
where $\{|i\rangle\}$ is a complete orthonormal basis of $A$. 
These states are conventionally called classical-quantum (CQ) states. Because of the nonconvexity of CQ, discord can be created just by mixing, and discord destroying maps cannot be linear.
Suppose the local density operator 
$\rho_A={\rm tr}_B\rho_{AB}$ admits a spectral decomposition $\rho_A=\sum_i p_i |i\rangle\langle i|$. Then
\begin{equation}
   \pi_A(\rho_{AB}) := \sum_i (|i\rangle_A\langle i|\otimes I_B)\rho_{AB}(|i\rangle_A\langle i|\otimes I_B),
\end{equation}
i.e., a local measurement in an eigenbasis of $A$,
is the most natural discord destroying map. 
Obviously, $\pi_A$ is nonlinear and thus not a channel: the basis in which the projection
takes place is dependent on the input state, and not uniquely defined within degenerate subspaces. Also note that $\pi_A$ never changes the local states.

We now plug $\pi_A$ into the conditions. Let $\mathcal{E}_A$ be a 
local operation acting on $A$. Note that we are considering the 
effect on the joint space: For example, the nongenerating 
condition reads 
$(\mathcal{E}_A\otimes I_B)\circ\pi_A=\pi_A\circ(\mathcal{E}_A\otimes I_B)\circ\pi_A$. 
This condition determines whether an operation always maps a 
CQ state to another.  As opposed to entanglement, discord 
can be created by certain local operations.
Such operations have been studied in Refs.\ \cite{hucpc,local}. 
$\bar{X}_A^\ast(\pi_A)$ and $X_A(\pi_A)$ have not been considered before to our knowledge.  
We can classify a variety of simple quantum operations according to
their behaviors in the theory of $\pi$ as follows (proofs in the Supplemental Material \cite{Note1}). Local unitary-isotropic channels (mixture of a unitary channel and depolarization, which are intuitively strongly discord-free) indeed belong to 
$X_A(\pi_A)$ and $X_{s,A}(\pi_A)$. Rank-one projective measurements, however, are in 
$\bar{X}_{s,A}(\pi_A)\backslash X_A(\pi_A)$. 
Furthermore, local mixed-unitary channels belong to all selective classes, but some of them are not in the original classes, supporting our general observation that selective conditions are not necessarily stronger than their original counterparts for nonlinear $\lambda$.  

As shown earlier, contractive distances between any 
$\rho_{AB}$ and $\pi_A(\rho_{AB})$, e.g., $S(\rho_{AB}||\pi_A(\rho_{AB}))$, is monotone under $X_A(\pi_A)$ (including all unitary-isotropic channels), and selectively monotone under $X_{s,A}(\pi_A)$ (including all mixed-unitary channels). 
This quantity is equivalent to a physically motivated 
simple measure of discord called diagonal discord \cite{lloydheat}.
(Similar quantities are independently discussed in 
Refs.\ \cite{rajagopal,mor,mid,unified,localdemon}.) 
Diagonal discord may suffer from discontinuities 
(infinitesimal perturbations may lead to a sudden jump 
in the value of diagonal discord) \cite{wpm,criteria}; 
however, it can be shown that they only occur at degeneracies \cite{zwl}.

Reference \cite{meznaric} adopts an approach similar to the idea of resource
destroying maps to study nonclassicality of operations. 
There, operations that commute with einselection 
\cite{zurekein} (complete dephasing) in a certain basis are regarded 
as classical.  
The key difference between the setup of Ref.\ \cite{meznaric} and 
the discord theory discussed here is that the basis for einselection 
needs to be specified; thus, not all discord-free states 
are the fixed points of such einselection \footnote{A simple consequence of this difference is that a local unitary can have infinite noncommutativity with the einselection (as measured by relative entropy) in the framework of Ref.\ \cite{meznaric}, as well as in our coherence theory, but it is always a $\pi$-commuting in our discord theory.}. Ref.\ \cite{meznaric} is more about local coherence in some preferred basis rather than discord.

\emph{Concluding remarks.|}In this Letter, we propose a simple framework 
for resource theories based on the notion of 
resource destroying maps. 
Our theory provides a general scheme for understanding the power of 
quantum operations in relation to certain resources.   The theory shows how to extend results that have been
previously derived for specific resources to a more general
class of resource theories.
In particular, our framework may lead to conceptual advances in 
understanding nonconvex theories such as discord.  
It would also be interesting to apply the framework of resource destroying
maps to other important resource theories, such as those of entanglement, magic states, 
asymmetry and thermodynamics.  

\emph{Note added.|}During the final revision of this Letter, we became aware of a recent review on discord \cite{adesso}, which includes a detailed discussion of the importance and difficulties of studying discord under the resource theory framework, and the state of the art of this field (in particular the local commutativity-preserving operations as the maximal set of local free operations).

\begin{acknowledgments}
{\it Acknowledgements.|}ZWL and SL are supported by AFOSR and ARO. XH is supported by NSFC under Grant No.\ 11504205.
We thank Can Gokler, Iman Marvian, Peter Shor, Kevin Thompson, Yechao Zhu, and anonymous referees for helpful discussions.
\end{acknowledgments}


%

\widetext
\appendix

\title{Supplemental Material}

\section{GENERAL ASPECTS OF RESOURCE DESTROYING MAPS}
Here we provide detailed discussions of general properties of the 
theory of resource destruction.  Specifically, we show that the 
convexity of the set of free states is a necessary condition for 
an associated resource destroying map to be linear, and analyse 
the robustness of the resource-free conditions.

\subsection{Linearity of resource destroying map}
We prove the following result that 
relates the convexity of the theory and the linearity of a resource destroying map:
\begin{thm}
Let $S$ be a set of states. Suppose $S$ is not convex, then there does not exist a linear map that stabilizes the states in $S$ only, i.e., no linear map $\lambda$ satisfies $\lambda(\rho)=\rho$ for all $\rho\in S$, and $\lambda(\sigma)\neq\sigma$ for all $\sigma\not\in S$.
\label{convex}
\end{thm}
\begin{proof}
By the nonconvexity of $S$, one can always find $\rho_1,\rho_2\in S$ such that $p\rho_1+(1-p)\rho_2\not\in S$ for some probability $p$. Suppose there is such a linear $\lambda$. Then $\lambda(p\rho_1+(1-p)\rho_2)=p\lambda(\rho_1)+(1-p)\lambda(\rho_2)=p\rho_1+(1-p)\rho_2$ by plugging in the property that $\lambda(\rho)=\rho$ when $\rho\in S$. This contradicts the other defining property of $\lambda$ that $\Lambda(\sigma)\neq\sigma$ when $\sigma\not\in S$.
\end{proof}
Recall that only free states are stabilized by a resource destroying map.
Therefore, for a theory with nonconvex $F$, no linear map can satisfy both requirements on all inputs. This implies that nonconvex theories do not admit resource destroying channels. 

\subsection{Robustness of resource-free conditions}
When the set of free states $F$ is not a singleton, i.e., contains more than one element, 
the definition of resource destroying map $\lambda$ is not unique. 
This leads to the question of whether different choices of $\lambda$ 
for a given $F$ define different resource nongenerating, 
nonactivating and commuting conditions (and their selective versions).

As mentioned in the main text, the resource nongenerating condition is robust: 
since $\lambda$ is always surjective onto $F$ by the nonresource-fixing 
requirement, this condition exclusively selects out the operations 
under which $F$ is closed. Here we show that, in contrast, resource 
nonactivating operations and thus commuting operations can depend on 
$\lambda$. 
For example, consider a peculiar $\lambda$ that maps all 
$\rho\not\in F$ to a particular $\rho_0\in F$. In other words, 
$\rho_0$ and all states outside $F$ form a family and all other 
free states are ``orphans'' without any parent states. If an 
operation does not stabilize all free states, then it has to 
map all states to the same image (a free state) to satisfy the 
nonactivating condition. Such a requirement is clearly stronger 
than the general case. 
An explicit example in the context of coherence is as follows. We define 
an extreme coherence-destroying map that takes all coherent states 
to one incoherent state $\rho_0$, while all other incoherent states are orphans.
Under this map, a partial depolarizing channel fails the 
nonactivating condition 
since it 
maps $\rho_0$ to $\tau\rho_0+(1-\tau)I/d$, which is still incoherent, 
while any coherent state remains coherent, thus always mapped to 
$\rho_0$. However, the partial depolarizing channel
obviously satisfies the (selective) nonactivating condition with $\Pi$. 
In addition, if we restrict the input to be free states, the 
resource nongenerating condition allows the same set 
of operations, but the nonactivating condition trivially holds for all
channels. Combining all these observations, we see that only 
the nongenerating condition is robust under the choice of
resource destroying map. 
Note that the above robustness results hold for the selective version of each condition.

It would be interesting to study resource non-activation
under appropriate restrictions on the choices for $\lambda$. 
For example, one could require that no orphans exist, 
or consider a class of resource destroying maps instead of a 
single one. As usual, we are primarily interested in physically 
motivated definitions. Despite the fragility of the 
resource nonactivating and commuting conditions under 
variations of resource destroying maps and input states, these 
conditions are still meaningful and nontrivial for physical 
definitions of $\lambda$.

\subsection{Selective monotonicity}

We say a resource measure $f(\rho)$ exhibits selective monotonicity 
if it is monotone nonincreasing under selective measurements on average. 
That is, $f(\rho)\geq \sum_\mu p_\mu f(\mathcal{E}_\mu(\rho))$ 
where $p_\mu = {\rm tr}\,\mathcal{E}_\mu(\rho)$.
Now consider a valid distance measure $D(\cdot,\cdot)$ that 
satisfies 
$D(\rho,\sigma)\geq \sum_\mu p_\mu 
D(\mathcal{E}_\mu(\rho),\mathcal{E}_\mu(\sigma))$, 
such as quantum relative entropy $S(\rho||\sigma)$ \cite{vedralplenio}. 
For selective commuting operations, i.e., $[\mathcal{E}_\mu,\lambda]=0$ for all $\mu$:
\begin{eqnarray}
 \tilde{\mathfrak{D}}(\rho) &:=& D(\rho,\lambda(\rho)) \nonumber\\
 &\geq& \sum_\mu p_\mu D(\mathcal{E}_\mu(\rho),\mathcal{E}_\mu(\lambda(\rho))) \nonumber\\
 &=& \sum_\mu p_\mu D(\mathcal{E}_\mu(\rho),\lambda(\mathcal{E}_\mu(\rho)))\nonumber\\
 &\equiv& \sum_\mu p_\mu \tilde{\mathfrak{D}}(\mathcal{E}_\mu(\rho)),
\end{eqnarray}
where the second line follows from the given property of $D$, and the third line follows from the selective commuting condition.
Thus, strong monotonicity holds for the simple measure $\tilde{\mathfrak{D}}$ introduced in Eq.\ (4) with proper $D$, e.g., the relative entropy between a state and its resource-destroyed counterpart, under selective commuting operations.

\section{COHERENCE}\label{appa}

Here we present a detailed analysis of the application of the theory of
resource destroying maps to the resource theory of coherence. 
We first analyze the comparative power of coherence-free classes defined by the theory of $\Pi$, namely $\bar{X}(\Pi)$, $\bar{X}^*(\Pi)$ and $X(\Pi)$ and their selective counterparts $\bar{X}_s(\Pi)$, $\bar{X}_s^*(\Pi)$ and $X_s(\Pi)$. We give some new examples of operations that exhibit characteristic behaviors in this theory.
Note that $\bar{X}_s(\Pi)$ and $X_s(\Pi)$ are respectively equivalent to Incoherent Operations (IO) \cite{baumgratz} and Strictly Incoherent Operations (SIO) \cite{sio} in literature. $\bar{X}(\Pi)$, namely coherence nongenerating operations, is recently studied by one of the authors \cite{hung}. In the preparation of this work, we became aware that $X(\Pi)$ is independently studied as Dephasing-covariant Incoherent Operations (DIO) in Refs.\ \cite{speakable,examination}. 
Notably, Ref.\ \cite{examination} argued that this class is a maximal extension of the ``physically consistent'' class, which is another interesting interpretation of $X(\Pi)$.
The structure of this theory is illustrated in Fig.\ \ref{cohfig}.
For comparison, we also briefly discuss about some coherence-free definitions arising from other scenarios, including Translationally Invariant Operations (TIO) \cite{time,speed} and Genuinely Incoherent Operations (GIO) \cite{genuine,powerg}. In this section, $\{|i\rangle\langle i|\}$ is the incoherent basis. 

\begin{figure}[H] 
\centering
\includegraphics[width=0.36\columnwidth]{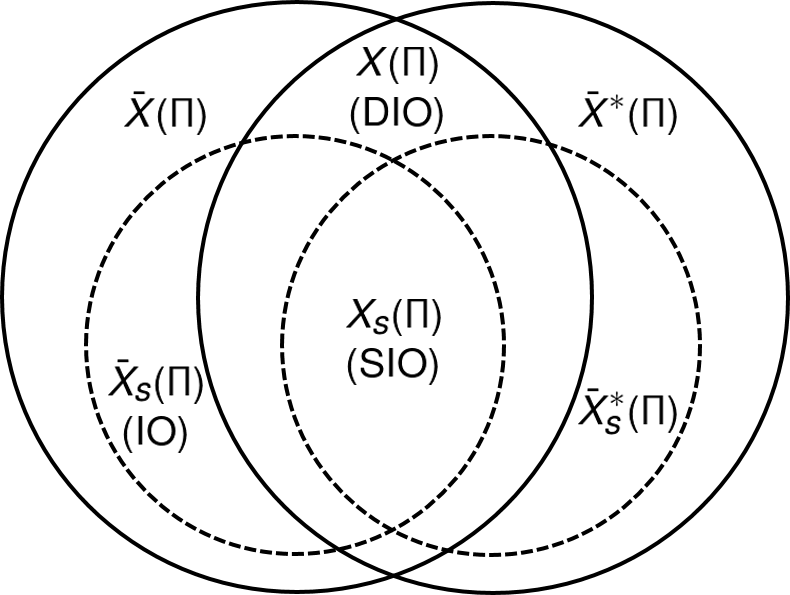}
\caption{\label{cohfig}A Venn diagram of coherence-free operations arising from the theory of resource destruction.}
\end{figure}

\subsection{$\bar{X}_s(\Pi)$ and $X(\Pi)$ are incomparable}
We first show that neither $\bar{X}_s(\Pi)$ and $X(\Pi)$ contains the other by constructing quantum operations belonging to $\bar{X}_s(\Pi)\backslash X(\Pi)$ and $X(\Pi)\backslash\bar{X}_s(\Pi)$. To achieve this, we derive the conditions on the entries of Kraus operators for operations in $\bar{X}_s(\Pi)$ and $X(\Pi)$. For an operation in $\bar{X}_s(\Pi)$, there is a Kraus decomposition such that 
\begin{equation}
   K^{(\mu)}_{ki}K^{(\mu)*}_{li}=0,\forall k\neq l
\end{equation}
for all Kraus operators $K^{(\mu)}$, where $K^{(\mu)}_{ki}=\langle{k}|K^{(\mu)}|{i}\rangle$. On the other hand, an operation in $X(\Pi)$ requires the summation \begin{equation}
   \label{Eq:Pi_Kraus}\sum_\mu K^{(\mu)}_{ki}K^{(\mu)*}_{li}=0,\forall k\neq l.
\end{equation}

Consider a qubit operation $\mathcal E_1(\rho)=K_1(\rho)K_1^\dagger+K_2(\rho)K_2^\dagger$ with $K_1=|0\rangle\langle +|$ and $K_2=|1\rangle\langle-|$. $K_1$ and $K_2$ are both incoherent so $\mathcal E_1\in\bar{X}_s(\Pi)$. However, it can be checked that $\mathcal E_1(\Pi(|+\rangle\langle+|))=I/2$ but $\Pi(\mathcal E_1(|+\rangle\langle+|))=|0\rangle\langle0|$, thus $\mathcal E_1\in \bar{X}_s(\Pi)\backslash X(\Pi)$. Since $\bar{X}_s(\Pi)\subset \bar{X}(\Pi)$, this operation is also an example of $\bar{X}(\Pi)\backslash X(\Pi)$.

Next consider a qutrit operation $\mathcal E_2$ with Kraus operators
\begin{equation}
\begin{array}{ccc}
 K_1=\left(\begin{array}{ccc}x_1&0&0\\
0&a&0\\
0&-b&0\end{array}\right),   & 

K_2=\left(\begin{array}{ccc}0&0&x_2\\
0&b^*&0\\
c^*&a^*&0\end{array}\right),   &
K_3=\left(\begin{array}{ccc}0&0&0\\
0&0&x_3\\
a&-c&0\end{array}\right),
\end{array}
\end{equation}
where the parameters satisfy $|x_1|^2+|c|^2+|a|^2=2|a|^2+2|b|^2+|c|^2=|x_2|^2+|x_3|^2=1$. It can be checked that any linear combination of the three Kraus operators is not incoherent. Since any other Kraus decomposition $\{M_i\}_{i=1}^d$ of $\mathcal{E}_2$ is related to $\{K_1,K_2,K_3\}$ by a $d$-dimensional unitary transformation $[u_{ij}]_{i,j=1}^d$ as $M_i=\sum_{j=1}^3u_{ij}K_j$, (and hence $M_i$ are not incoherent), we conclude that $\mathcal E_2\not\in\bar{X}_s(\Pi)$. Meanwhile, $\{K_1,K_2,K_3\}$ satisfy Eq.\ (\ref{Eq:Pi_Kraus}), thus $\mathcal E_2\in X({\Pi})\backslash\bar{X}_s(\Pi)$.

We also note that $\bar{X}_s(\Pi)\cap X(\Pi)$ is not empty, because $X_s(\Pi)$ (studied in \cite{sio}) is a subset of both. So $\bar{X}_s(\Pi)$ and $X(\Pi)$ are incomparable but not disjoint.

\subsection{Incoherent-measure-and-prepare operations are in $\bar{X}_s^\ast(\Pi)$}
By previous results and references cited in the introduction, 
we already have a full characterization of the coherence 
nongenerating and $\Pi$-commuting classes. The question
of non-activation has not been studied before, however.
Accordingly, we exhibit here a class of operations belonging 
to $\bar{X}_s^\ast(\Pi)$ (and thus $\bar{X}^\ast(\Pi)$). Some of them are able to generate coherence, so we also have examples of $\bar{X}_s^\ast(\Pi)\backslash X(\Pi)$.

Consider the following type of operations such that the 
Kraus operators take the form $K_i=|f_i\rangle\langle i|$. 
Such operations represent the following measure-and-prepare procedure: 
one first performs a projective measurement in the incoherent basis, 
and then prepares the system the corresponding state $|f_i\rangle$ 
upon measuring $i$. Such operations are entanglement 
breaking when acting locally \cite{entbreak}.
Here we show that they belong to $\bar{X}_s^\ast(\Pi)$.  Let $\mathcal{E}^{\rm MP}$ be such an operation. 
Notice that the measuring step of $\mathcal{E}^{\rm MP}$ destroys coherence, so $\mathcal{E}^{\rm MP}\circ\Pi=\mathcal{E}^{\rm MP}$ automatically holds. More explicitly, for any $\rho$,
\begin{eqnarray}
 \mathcal{E}^{\rm MP}(\Pi(\rho))&=&\sum_{ij} |f_i\rangle\langle i|j\rangle\langle j|\rho|j\rangle\langle j|i\rangle\langle f_i| \nonumber\\
 &=& \sum_i |f_i\rangle\langle i|\rho|i\rangle\langle f_i| = \mathcal{E}^{\rm MP}(\rho),
\end{eqnarray}
where we used $\langle i|j \rangle=\delta_{ij}$ for the second equality.
The same condition also holds for each Kraus arm. 
So $\mathcal{E}^{\rm MP}(\rho)\in\bar{X}_s^\ast(\Pi)$. 
Now suppose there exists some $i$ such that $|f_i\rangle$ is coherent: then the operation ceases to be coherence nongenerating (simply take $|i\rangle$ as the input). Therefore, such operations reside in $\bar{X}_s^\ast(\Pi)\backslash X(\Pi)$.

\subsection{Other definitions of coherence-free operations}
Besides those derived from coherence destruction $\Pi$ as shown above, there are also some other proposals of coherence-free operations arising from different contexts.
Two notable ones are Translationally Invariant Operations (TIO) \cite{time,speed} and Genuinely Incoherent Operations (GIO) \cite{genuine,powerg}. However, we argue that both of them are not theories of coherence with respect to a specified observable in a precise sense. 

TIO naturally arises from the asymmetry theory \cite{refrmp,asymgour}, since coherence can be viewed as asymmetry relative to time translations generated by some preferred Hamiltonian \cite{time,speed}.
An operation $\mathcal E^{\rm TI}$ is said to be translationally-invariant with respect to a Hamiltonian $H$ if it satisfies
\begin{equation}
   \mathcal E^{\rm TI}(e^{-iHt}\rho e^{iHt})=e^{-iHt}\mathcal E^{\rm TI}(\rho) e^{iHt},\forall t,
\end{equation}
for any state $\rho$. It turns out that the power of TIO depends on whether $H$ exhibits degeneracy or not. For general $H$, it is possible to generate coherence within the decoherence-free subspaces using TIO, so this class is technically not even contained in the maximal class $\bar{X}(\Pi)$. However, when $H$ has a nondegenerate spectrum, the resulting class ${\rm TIO}^\ast$) defines more precisely a theory of coherence with respect to the eigenbasis of $H$.  It can be shown that ${\rm TIO}^\ast\subset X_s(\Pi)$ \cite{sio} (earlier Ref.\ \cite{speed} showed that ${\rm TIO}^\ast\subset \bar{X}_s(\Pi)$).

The concept of GIO is proposed in order to remove the dependence of incoherence on specific experimental realizations. In other words, the Kraus arms are unable to create coherence for \emph{all} Kraus decompositions of a GIO, or equivalently, all Kraus operators are diagonal in the incoherent basis \cite{genuine}.
A consequence is that all incoherent states are invariant under GIO.
Therefore, in some sense, GIO represents a theory with more constraints in addition to those imposed by incoherence since it cannot even achieve transformations among incoherent (free) states.
Indeed, it is known that $\mathrm{GIO}\subset \mathrm{TIO}^\ast$ \cite{speed} and $\mathrm{TIO}^\ast\subset X_s(\Pi)$ \cite{sio}: $\mathrm{GIO}$ is strictly weaker than the weakest class given by the theory of $\Pi$. We include a more intuitive and straightforward proof of $\mathrm{GIO}\subset X_s(\Pi)$ here.
Let $\mathcal{E}^{\rm GI}$ be a GIO.  
Suppose a Kraus decomposition of $\mathcal{E}^{\rm GI}$ reads $\mathcal{E}^{\rm GI}(\rho)=\sum_\mu K_\mu \rho K^\dagger_\mu$ and $\mathcal{E}^{\rm GI}_\mu(\rho):=K_\mu \rho K^\dagger_\mu$. By Theorem 1 of Ref.\ \cite{genuine}, $K_\mu$ is diagonal in the incoherent basis. In other words, $K_\mu|i\rangle=\alpha_{i\mu}|i\rangle$ for all $i,\mu$, where $\alpha_{il}$ is some constant.
 Notice that any state $\rho$ can be written in the form of
\begin{equation}
   \rho=\Pi(\rho)+\sum_{j \neq i}|i\rangle\langle i|\rho|j\rangle\langle j|,
\end{equation} 
which separates the diagonal and off-diagonal parts. By linearity of $\mathcal{E}^{\rm GI}_\mu$,
\begin{eqnarray}
   \mathcal{E}^{\rm GI}_\mu(\rho)&=&\mathcal{E}^{\rm GI}_\mu(\Pi(\rho))+\sum_{j \neq i} K_\mu|i\rangle\langle i|\rho|j\rangle\langle j|K^\dagger_\mu\nonumber\\
   &=&\mathcal{E}^{\rm GI}_\mu(\Pi(\rho))+\sum_{j \neq i} \alpha_{i\mu}\alpha_{j\mu}^\ast|i\rangle\langle i|\rho|j\rangle\langle j|.\label{lin}
\end{eqnarray} 
Notice that the off-diagonal parts remain off-diagonal, which are erased by a following $\Pi$. That is, $\Pi(\mathcal{E}^{\rm GI}_\mu(\rho))=\Pi(\rho)=\mathcal{E}^{\rm GI}_\mu(\Pi(\rho))$ for all $\rho$, i.e., $\mathcal{E}^{\rm GI}_\mu\circ\Pi=\Pi\circ\mathcal{E}^{\rm GI}_\mu$.
So $\mathrm{GIO}\subseteq X_s(\Pi)$.
To see that the containment is proper, consider an erasure channel that maps everything to $|0\rangle\langle 0|$, where $|0\rangle$ is an incoherent basis state. This channel obviously belongs to $X_s(\Pi)$. But it is not a GIO since it does not leave incoherent states invariant except $|0\rangle\langle 0|$.

\section{DISCORD}\label{appb}

It is much more difficult to study discord-free conditions since the discord destroying map $\pi$ (which acts locally) is dependent on the input and not uniquely defined within degenerate subspaces.
To obtain some preliminary understandings of this $\pi$ theory, we examine the power of some of the most typical quantum operations acting locally on the same subsystem as $\pi$ ( without loss of generality, subsystem $A$). 
We show that local unitary-isotropic channels (mixture of some unitary channel and depolarization, or unitary with white noise) exhibit the strongest classicality: they belong to both $X_A(\pi_A)$ and $X_{s,A}(\pi_A)$. Nevertheless, rank-one projective measurements fail to be nonactivating: they reside in $\bar{X}_{s,A}(\pi_A)\backslash X_A(\pi_A)$. In addition, a peculiar feature that distinguishes nonlinear resource destroying maps from linear ones is that selective classes are not necessarily contained in their original counterparts. We confirm this in the $\pi$ theory by showing that $\bar{X}_{s,A}(\pi_A)\backslash \bar{X}_A(\pi_A)$ contain certain qudit ($d>2$) mixed-unitary channels.
We also provide a measure-and-prepare protocol that is able to generate but not activate discord. However, this protocol does not represent a channel since it depends on the eigenbasis of the input.

In this section, we follow the notations used in the main text: $\rho_{AB}$ is an arbitrary bipartite state, $\rho_A=\sum_i p_i|i\rangle\langle i|$ ($\{|i\rangle\langle i|\}$ diagonalizes the reduced density operator of $A$).

\subsection{Unitary-isotropic channels are in $X_A(\pi_A)$ and $X_{s,A}(\pi_A)$}
Unitary-isotropic channels take the form $\tilde{u}^\gamma(\rho)=(1-\gamma)U\rho U^\dagger+\gamma I/d$, where $U$ is unitary, $\gamma$ characterizes the degree of depolarization ($\gamma\in[0,d^2/(d^2-1)]$ so that $\tilde{u}^\gamma$ is completely positive \cite{2016arXiv161007504B}), and $d$ is the dimension of the Hilbert space.
Unitary channels ($\gamma=0$) and depolarizing channels ($U=I$) are special cases of unitary-isotropic channels.

On the one hand, $\pi_A$ is a local measurement in the $\{|i\rangle\langle i|\}$ basis, so
\begin{eqnarray}
\pi_A(\rho_{AB}) &=& \sum_{i} |i\rangle_A\langle i| \rho_{AB} |i\rangle_A\langle i|,\\
   (\tilde{u}^\gamma_A \otimes I_B)(\pi_A(\rho_{AB}))&=&(1-\gamma)U_A\left(\sum_{i} |i\rangle_A\langle i| \rho_{AB} |i\rangle_A\langle i|\right) U^\dagger_A +  \frac{\gamma}{d_A}I_A\otimes\rho_B\nonumber\\
   &=& (1-\gamma)\sum_i U_A|i\rangle_A\langle i|U_A^\dagger\otimes\langle i|\rho_{AB}|i\rangle+\frac{\gamma}{d_A}I_A\otimes\rho_B,
\end{eqnarray}
where ${\rm tr}\langle i|\rho_{AB}|i\rangle=p_i$.

On the other hand,
 \begin{equation}
    (\tilde{u}^\gamma_A\otimes I_B)(\rho_{AB})=(1-\gamma)(U_A\otimes I_B)\rho_{AB}(U_A^\dagger\otimes I_B)+\frac{\gamma}{d_A} I_A\otimes\rho_B.
 \end{equation}
 Notice that ${\rm tr}_B(\tilde{u}^\gamma_A \otimes I_B)(\rho_{AB})=\sum_i ((1-\gamma)p_i+\gamma/d) U_A|i\rangle_A\langle i|U_A^\dagger $, so $\{U|i\rangle\langle i| U^\dagger\}$ is a new eigenbasis, which implies that
  \begin{eqnarray}
    \pi_A((\tilde{u}^\gamma_A\otimes I_B)\rho_{AB})&=&(1-\gamma)\sum_i U_A|i\rangle_A\langle i|U_A^\dagger U_A\rho_{AB}U_A^\dagger U_A|i\rangle_A\langle i|U_A^\dagger+\frac{\gamma}{d_A}\sum_i U_A|i\rangle_A\langle i|U_A^\dagger I_A U_A|i\rangle_A\langle i|U_A^\dagger\otimes\rho_B\nonumber\\
    &=& (1-\gamma)\sum_i U_A|i\rangle_A\langle i|U_A^\dagger\otimes\langle i|\rho_{AB}|i\rangle+\frac{\gamma}{d_A}I_A\otimes\rho_B\nonumber\\
    &=&(\tilde{u}^\gamma_A\otimes I_B)(\pi_A(\rho_{AB})).
 \end{eqnarray}
 That is, $(\tilde{u}^\gamma_A\otimes I_B)\circ \pi_A =\pi_A \circ (\tilde{u}^\gamma_A\otimes I_B)$: $\tilde{u}^\gamma\in X_A(\pi_A)$.

 Now notice that applying Heisenberg-Weyl operators uniformly at random (Heisenberg-Weyl twirling) on any qudit gives the maximally mixed state \cite{wilde}. This indicates that unitary-isotropic channels admit a Kraus decomposition by unitaries (belong to mixed-unitary channels). More explicitly,
 \begin{equation}
     \tilde{u}^\gamma(\rho)=(1-\gamma)U\rho U^\dagger +\frac{\gamma}{d^2}\sum_{i,j=0}^{d-1}\mathsf{X}^i\mathsf{Z}^j\rho{\mathsf{Z}^\dagger}^j{\mathsf{X}^\dagger}^i \equiv \left(1-\gamma\frac{d^2-1}{d^2}\right)U\rho U^\dagger + \frac{\gamma}{d^2}\sum_{\substack{i,j=0\\i+j\neq 0}}^{d-1}U\mathsf{X}^i\mathsf{Z}^j\rho{\mathsf{Z}^\dagger}^j{\mathsf{X}^\dagger}^i U^\dagger,
 \end{equation}
 where $\mathsf{X}$ and $\mathsf{Z}$ are generalized Pauli operators acting unitarily as $\mathsf{X}|j\rangle=|{j+1}\mod d\rangle$ (cyclic shift) and $\mathsf{Z}|j\rangle=e^{i2\pi j/d}|j\rangle$ (phase). $\mathsf{X}^i\mathsf{Z}^j$ for $i,j=0,\cdots, d-1$ are Heisenberg-Weyl operators. Since unitaries belong to $X_A(\pi_A)$, we conclude that $\tilde{u}^\gamma\in X_{s,A}(\pi_A)$.

\subsection{Rank-one projective measurements are in $\bar{X}_{s,A}(\pi_A)\backslash X_A(\pi_A)$}

Now consider a local projective measurement in the basis $\{|\psi_j\rangle\langle\psi_j|\}$, denoted by $\Psi$. This operation is obviously commutativity-preserving (thus in $\bar{X}({\pi})$) since its output is always diagonal in the specified basis. In fact, each projection is trivially commutativity-preserving, so $\Psi\in\bar{X}_{s,A}(\pi_A)$.
    
    Then we consider if $\Psi_A$ always commutes with $\pi_A$. On the one hand,
    \begin{equation}
    (\Psi_A\otimes I_B)(\pi_A(\rho_{AB}))=
    \sum_{i,j}|\psi_j\rangle_A\langle\psi_j|\otimes\langle\psi_j|i\rangle_A\langle{i}|\rho_{AB}|i\rangle_A\langle{i}|\psi_j\rangle.
    \label{d1}
    \end{equation}
   On the other hand,
    \begin{equation}
      (\Psi_A\otimes I_B)(\rho_{AB})= \sum_{j}|\psi_j\rangle_A\langle\psi_j|\otimes\langle\psi_j|\rho_{AB}|\psi_j\rangle,\label{d2}
    \end{equation}
    which is already classical-quantum, so $\pi_A((\Psi_A\otimes I_B)(\rho_{AB}))=(\Psi_A\otimes I_B)(\rho_{AB})$. 
    Since $\rho_{AB}$ is arbitrary, the right hand sides of Eqs.\ (\ref{d1}) and (\ref{d2}) always coincide if and only if $\{|i\rangle\}=\{|\psi_i\rangle\}$. (The equality always holds if we restrict to classical-quantum inputs so that $|i\rangle_A\langle{i}|\otimes I_B$ and $\rho_{AB}$ commute.)
    This indicates that $\Psi_A$ do not commute with $\pi_A$ when $\Psi_A$ is not an eigenbasis of $\rho_A$. Therefore, $\Psi\in\bar{X}_{s,A}(\pi_A)\backslash X_A(\pi_A)$. 
    
\subsection{Some mixed-unitary channels are in $\bar{X}_{s,A}(\pi_A)\backslash \bar{X}_A(\pi_A)$}   
   By previous results, mixed-unitary channels belong to all selective classes. Here we argue that certain mixed-unitary channels live outside $\bar{X}_A(\pi_A)$, thus confirming that selective conditions are not necessarily stronger than their original versions in theories with nonlinear resource destroying maps.
   
   
   When $A$ is a qubit, $\bar{X}_A(\pi_A)$ (commutativity-preserving channels) is composed of unital channels, which is known to be equivalent to mixed-unitary channels (qubit quantum Birkhoff theorem) \cite{watrous}, and semiclassical (SC) channels (the outputs of a semiclassical channel are diagonal in the same basis) \cite{local,hucpc}. Therefore, the following classes collapse:
   \begin{equation}
      \bar{X}_A(\pi_A)\backslash{\rm SC} = \text{Mixed-unitary} = \text{Unital}.
   \end{equation}
   
   When $A$ is a qudit with dimension $d>2$, however, the above classes form a strict hierarchy. In this case, $\bar{X}_A(\pi_A)\backslash{\rm SC}$ are composed of isotropic channels (which takes the form $(1-\gamma)\Gamma({\rho})+\gamma I/d$, where $\Gamma$ is either unitary or antiunitary) \cite{hucpc,guohou}, which belong to mixed-unitary channels (unitary-isotropic case: by previous results; antiunitary-isotropic case: Ref.\ \cite{2016arXiv161007504B}). However, since $\Gamma$ preserves eigenvalues, $\mu$ must isotropically ``shrink'' the spectrum towards uniformity by degree $\lambda$ (which does not distinguish between reference frames, as its name suggests).  For example, given any two pure states as inputs, the respective outputs of an isotropic channel must have the same spectrum.  Therefore, the property of isotropy places a strong uniformity requirement on the mixtures of unitaries.  
    An example of anisotropic mixed-unitary channels is given as follows. Let $\mu(\rho)=\mathsf{X}\rho\mathsf{X^\dagger}/2+\mathsf{X^2}\rho\mathsf{{X^\dagger}^2}/2$ be a qutrit mixed-unitary channel, where $\mathsf{X}$ is the cyclic shift generalized Pauli operator as defined earlier. 
   Define $|\chi\rangle:=(|0\rangle+|1\rangle+|2\rangle)/\sqrt{3}$, which is an eigenstate of both $\mathsf{X}$ and $\mathsf{X}^2$. Then
   \begin{equation}
    \mu(|\chi\rangle\langle\chi|)=|\chi\rangle\langle\chi|,
    \end{equation}
    which remains a pure state, but
    \begin{equation}
    \mu(|0\rangle\langle 0|)=\frac{1}{2}|1\rangle\langle 1|+\frac{1}{2}|2\rangle\langle 2|,
   \end{equation}
   which is mixed. So $\mu$ is not an isotropic channel. Moreover, it is well known that there exist qudit unital channels that is not mixed-unitary (no quantum Birkhoff theorem) \cite{watrous,shorppt}. In conclusion, for $d>2$,
      \begin{equation}
      \bar{X}_A(\pi_A)\backslash{\rm SC} \subsetneq \text{Mixed-unitary} \subsetneq \text{Unital}.
   \end{equation}
Since $\mu(|\chi\rangle\langle\chi|)$ and $\mu(|0\rangle\langle 0|)$ clearly do not commute, $\mu$ is also not semiclassical. So $\mu\not\in\bar{X}_A(\pi_A)$.
 
\subsection{A measure-and-prepare map}

Lastly, we define the following measure-and-prepare protocol that generates discord, but cannot activate it. Given an input $\rho_{AB}$, consider the measure-and-prepare operation with Kraus operators $K_i=|g_i\rangle\langle i|$ (recall that $\{|i\rangle\}$ diagonalizes $\rho_A$), but there exists $j\neq k$ such that $\langle g_j|g_k\rangle\neq0$. Let $\xi$ be such a measure-and-prepare map.  By definition, it is not commutativity-preserving, thus able to create discord. However, for any $\rho_{AB}$,
\begin{eqnarray}
 (\xi_A\otimes I_B)(\pi_A(\rho_{AB}))&=&\sum_{ij} |g_i\rangle_A\langle i|j\rangle_A\langle j|\rho_{AB}|j\rangle_A\langle j|i\rangle_A\langle g_i| \nonumber\\
 &=& \sum_i |g_i\rangle_A\langle i|\rho_{AB}|i\rangle_A\langle g_i| =  (\xi_A\otimes I_B)(\rho_{AB}),
\end{eqnarray}
where we used $\langle i|j\rangle=\delta_{ij}$ for the second equality.
That is, $\xi_A\otimes I_B=(\xi_A\otimes I_B)\circ\pi_A$. So $\xi$ is nonactivating.
Note that the above protocol is not linear.
It remains an open question as to whether there are quantum channels in $\bar{X}^\ast_A(\pi_A)\backslash X_A(\pi_A)$.

\end{document}